\documentclass[letterpaper, 10 pt, conference]{ieeeconf}  

\IEEEoverridecommandlockouts                              
\usepackage{amsmath} 
\usepackage{amssymb}  
\usepackage{mathrsfs}
\usepackage{mathtools}
\usepackage{bm}
\usepackage{physics}
\usepackage{cite}
\usepackage{graphicx}
\usepackage{subfigure}

\usepackage{algorithm}
\usepackage[noend]{algpseudocode}

\usepackage{etoolbox}
\makeatletter
\patchcmd{\ALG@doentity}{\item[]\nointerlineskip}{}{}{}
\makeatother

\algrenewcommand\algorithmicindent{1em}

\makeatletter
\renewcommand{\ALG@name}{Protocol}
\makeatother
\algrenewcommand\algorithmicdo{}

\newcommand{\defref}[1]{Definition~\ref{#1}}
\newcommand{\asmref}[1]{Assumption~\ref{#1}}
\newcommand{\propref}[1]{Proposition~\ref{#1}}
\newcommand{\lemref}[1]{Lemma~\ref{#1}}
\newcommand{\thmref}[1]{Theorem~\ref{#1}}

\newcommand{\secref}[1]{Section~\ref{#1}}
\newcommand{\figref}[1]{Fig.~\ref{#1}}

\newcommand{\proref}[1]{Protocol~\ref{#1}}

\newcommand{\N}{\mathbb{N}}
\newcommand{\Z}{\mathbb{Z}}
\newcommand{\R}{\mathbb{R}}

\newcommand{\M}{\mathcal{M}}
\newcommand{\C}{\mathcal{C}}

\newcommand{\Enc}{\mathsf{Enc}}
\newcommand{\Dec}{\mathsf{Dec}}

\newcommand{\Share}{\mathsf{Share}}
\newcommand{\Reconst}{\mathsf{Reconst}}

\newcommand{\pk}{\mathsf{pk}}
\newcommand{\sk}{\mathsf{sk}}
\newcommand{\ct}{\mathsf{ct}}

\newcommand{\view}[1]{\mathsf{view}^{#1}}
\newcommand{\Sim}{\mathsf{Sim}}

\newcommand{\xx}{\mathrm{x}}
\newcommand{\zz}{\mathrm{z}}
\newcommand{\uu}{\mathrm{u}}
\newcommand{\vv}{\mathrm{v}}
\newcommand{\ww}{\mathrm{w}}
\newcommand{\sss}{\mathrm{s}}
\newcommand{\trm}{\mathrm{t}}
\newcommand{\xxsf}{\mathsf{x}}
\newcommand{\yysf}{\mathsf{y}}
\newcommand{\mmsf}{\mathsf{m}}
\newcommand{\rrsf}{\mathsf{r}}

\newcommand{\Nin}{N^{\mathrm{in}}}
\newcommand{\Nout}{N^{\mathrm{out}}}

\newtheorem{definition}{Definition}
\newtheorem{assumption}{Assumption}
\newtheorem{proposition}{Proposition}
\newtheorem{lemma}{Lemma}
\newtheorem{theorem}{Theorem}

\newtheorem{remark}{Remark}

\allowdisplaybreaks[4]

\title{\LARGE \bf
    Faithful and Privacy-Preserving Implementation of Average Consensus$^\ast$
}

\author{Kaoru Teranishi$^{1,2}$, Kiminao Kogiso$^{3}$, and Takashi Tanaka$^{4}$
\thanks{$^{\ast}$This work was supported by JSPS Grant-in-Aid for JSPS Fellows Grant Number JP21J22442 and for JSPS KAKENHI Grant Number JP23K22779.}
\thanks{$^{1}$School of Aeronautics and Astronautics,
        Purdue University, West Lafayette, IN 47907, USA
        {\tt\small kteranis@purdue.edu}}%
\thanks{$^{2}$Japan Society for the Promotion of Science, Chiyoda, Tokyo, Japan}%
\thanks{$^{3}$Department of Mechanical and Intelligent Systems Engineering,
        The University of Electro-Communications, 1-5-1 Chofugaoka, Chofu, Tokyo 182-8585, Japan
        {\tt\small kogiso@uec.ac.jp}}%
\thanks{$^{4}$School of Aeronautics and Astronautics, Elmore Family School of Electrical and Computer Engineering,
        Purdue University, West Lafayette, IN 47907, USA
        {\tt\small tanaka16@purdue.edu}}%
}

\begin{document}

\thispagestyle{empty}
\hspace{-4.5mm}
\fbox{
\begin{minipage}{\textwidth-5mm}\scriptsize
© 20XX IEEE.
Personal use of this material is permitted.
Permission from IEEE must be obtained for all other uses, in any current or future media, including reprinting/republishing this material for advertising or promotional purposes, creating new collective works, for resale or redistribution to servers or lists, or reuse of any copyrighted component of this work in other works.
\end{minipage}
}
\newpage
\setcounter{page}{0}

\maketitle
\thispagestyle{empty}
\pagestyle{empty}

\begin{abstract}
    We propose a protocol based on mechanism design theory and encrypted control to solve average consensus problems among rational and strategic agents while preserving their privacy.
    The proposed protocol provides a mechanism that incentivizes the agents to faithfully implement the intended behavior specified in the protocol.
    Furthermore, the protocol runs over encrypted data using homomorphic encryption and secret sharing to protect the privacy of agents.
    We also analyze the security of the proposed protocol using a simulation paradigm in secure multi-party computation.
    The proposed protocol demonstrates that mechanism design and encrypted control can complement each other to achieve security under rational adversaries.
\end{abstract}

\section{Introduction}
\label{sec:introduction}

Average consensus is a fundamental problem in multi-agent systems to reach an agreement on the average of agents' states.
It arises in numerous applications, such as rendezvous of mobile robots, data fusion in sensor networks, and distributed optimization~\cite{Ren2005-td,Olfati-Saber2007-ny,Kia2019-lq}.
This problem is usually solved by exchanging information among agents and updating their states locally based on the information when they are cooperative.
However, a rational and strategic agent may be incentivized to manipulate the average consensus algorithm (e.g., by misreporting information) to drive an outcome to its own benefit.
Furthermore, adversarial agents may learn the secrets of honest agents through information exchanges, thereby compromising their privacy.

We address these challenges by designing a protocol that combines mechanism design and encrypted control.
Mechanism design theory deals with the design of rules to achieve preferable social outcomes in the presence of strategic agents~\cite{Shoham2008-kd}.
In classical mechanism design, a social planner asks agents to report their private information and announces a social decision and tax computed using the collected information.
In contrast, distributed mechanism design considers determining the outcome in a distributed manner~\cite{Parkes2004-zk}.
Previous studies~\cite{Tanaka2013-bl,Tanaka2017-zg} have shown that some distributed optimization and control algorithms can be faithfully implemented using distributed mechanisms.

Encrypted control is a framework that applies cryptographic primitives to decision-making in dynamical systems~\cite{Schulze_Darup2021-qq}.
Within this framework, previous studies considered consensus control~\cite{Kishida2018-ws,Ruan2019-gz}, formation control~\cite{Marcantoni2023-lx}, and cooperative control~\cite{Schulze_Darup2019-kw,Alexandru2019-dn} using homomorphic encryption and secret sharing.
These cryptographic primitives enable the computation of sensitive information in an encrypted form.
Thus, encrypted control is effective in mitigating privacy compromises in multi-agent systems.

Although mechanism design and encrypted control have been developed individually so far, integrating these methodologies would produce a promising approach to simultaneously achieve both faithfulness and privacy in cooperative decision-making.
Traditional mechanism design based on the \emph{revelation principle}~\cite{Shoham2008-kd} to attain faithful implementation requires agents to disclose their private information to a social planner.
This process is clearly undesirable from a privacy perspective and will be improved by running the computation of mechanisms over encrypted data.
On the other hand, existing encrypted controls assume semi-honest agents who may attempt to learn private information from received messages but do not deviate from a protocol, thus failing to address strategic manipulation by agents.
A mechanism can dissuade strategic agents from manipulating a protocol, which is expected to enhance the achievable security of encrypted controls.

The main contribution of this study lies in clarifying the synergy between mechanism design and encrypted control.
We demonstrate how incentivization by mechanism design and secure computation using cryptographic primitives can achieve average consensus in the presence of rational agents rather than semi-honest agents.
Specifically, our contributions are listed as follows.
1) We propose a privacy-preserving protocol to provide a mechanism that implements an average consensus by adopting the algorithm in~\cite{Tanaka2013-bl} tailored for multi-party computation.
In contrast to the previous algorithm, the mechanism computation in the proposed protocol is distributedly performed by the agents instead of a single leader.
2) Building on previous studies~\cite{Tanaka2013-bl,Tanaka2017-zg}, we show that the agents do not deviate from the intended behavior even though their private information reports are encrypted.
3) We also demonstrate that the proposed protocol fulfills a standard privacy requirement for secure multi-party computation through simulation-based proofs~\cite{Lindell2016-gz,Goldreich2009-ct}.

The remainder of this paper is organized as follows.
\secref{sec:problem} describes a problem setting.
\secref{sec:preliminaries} introduces definitions of mechanism design and secure multi-party computation.
\secref{sec:consensus} presents a solution to the problem and demonstrates its security under semi-honest adversaries.
\secref{sec:mechanism} proposes a protocol providing a mechanism that implements the solution under rational adversaries.
\secref{sec:examples} illustrates the effectiveness of the proposed protocol through numerical simulations.
\secref{sec:conclusions} describes the conclusions of this study.

\section{Problem Setting}
\label{sec:problem}

In this study, we consider the multi-agent system that consists of $N$ agents shown in \figref{fig:mas}.
Each agent is directly connected to the supervisor, whose role will be explained later.
The network topology of the agents is described by a strongly connected and balanced digraph $G = (V, E)$ with a vertex set $V = \{1, \dots, N\}$ and edge set $E \subset V \times V$.
A weighted adjacency matrix of $G$ is $A = [a_{ij}] \in \R^{N \times N}$, where $a_{ij} > 0$ if $(i, j) \in E$ and $a_{ij} = 0$ otherwise.
Define the input and output neighbors of agent~$i$ as $\Nin_i \coloneqq \{ j \in V \mid (i, j) \in E \}$ and $\Nout_i \coloneqq \{ j \in V \mid (j, i) \in E \}$, respectively.
Here, $(i, j) \in E$ means that agent~$j$ can send a message to agent~$i$.
The dynamics of agent~$i$ is given by
\begin{equation}
    \xx_i(k + 1) = \xx_i(k) + \uu_i(k),
    \label{eq:agent}
\end{equation}
where $k \in \N_0 \coloneqq \{ 0, 1, 2, \dots \}$ is the time index, $\xx_i(k) \in \R$ is the state, $\uu_i(k) \in \R$ is the input, and the initial state is $\xx_i(0) = \xx_{i, 0}$.
A common average-consensus control for \eqref{eq:agent} is
\begin{equation}
    \uu_i(k) = \ww_{ii} \xx_i(k) + \sum{}_{j \in \Nin_i} \ww_{ij} \xx_j(k),
    \label{eq:control}
\end{equation}
where $\ww_{ij} = \epsilon a_{ij}$ for $i \ne j$, $\ww_{ii} = -\sum_{j \in \Nin_i} \ww_{ij}$, and $\epsilon \in (0, 1 / \max_i \sum_{j \ne i} a_{ij})$.
Applying \eqref{eq:control} to \eqref{eq:agent}, the system achieves average consensus, i.e., $\lim_{k \to \infty} \xx_i(k) = \frac{1}{N} \sum_{j=1}^N \xx_{j,0}$ for all $i \in V$, if $G$ is strongly connected and balanced~\cite{Olfati-Saber2007-ny}.

\begin{figure}[t]
    \centering
    \includegraphics[scale=1]{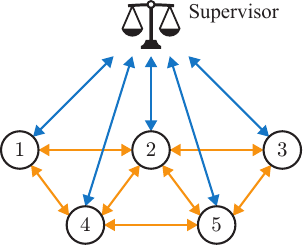}
    \caption{Multi-agent system with a supervisor ($N = 5$).}
    \label{fig:mas}
    \vspace{-5mm}
\end{figure}

The goal of this study is to design a secure multi-party computation protocol that achieves average consensus under the following computation and security models.

\paragraph{Computation model}
We employ a preprocessing model~\cite{Damgard2012-jh} to simplify a protocol and to improve its efficiency.
A protocol in this model consists of offline and online phases.
In the offline phase, agents receive auxiliary inputs from a third party independently of their inputs.
In the online phase, the agents compute outputs using the auxiliary inputs without the third party.
The supervisor in \figref{fig:mas} plays the role of both the third party in this model and a social planner in mechanism design.
More precisely, the supervisor distributes random numbers to the agents in the offline phase and verifies tax payments once the protocol is terminated.

\paragraph{Security model}
We propose a secure protocol under rational adversaries.
Unlike a semi-honest adversary, a rational adversary does not only attempt to compromise the privacy of honest agents but also acts strategically to minimize its own cost.
Such an adversary may rationally deviate from the designated protocol if the protocol is not \emph{incentive compatible} (defined below).
The proposed protocol aims to prevent rational adversaries from learning information beyond their inputs and outputs while attaining incentive compatibility.

\begin{remark}
    The supervisor is a third party independent from the multi-agent system.
    It provides the moderation and coordination service for the agents to faithfully and privately realize an average consensus task.
    In this scenario, the tax payment can be thought of as the fee that each agent is asked to pay for the service, which is ensured by contract. 
\end{remark}

\section{Mechanism Design and Multi-party Computation}
\label{sec:preliminaries}

\subsection{Mechanism design}

The notation used in this section is consistent with~\cite{Tanaka2017-zg}.
Suppose each agent~$i \in \{1, \dots, N\}$ has private information (called type) $\theta_i \in \Theta_i$, which defines its cost $u_i(d(\theta), t(\theta); \theta_i) = v_i(d(\theta); \theta_i) + t_i(\theta)$ for a decision rule $d: \Theta \to X$ and transfer rule $t: \Theta \to \R^N$, where $t(\theta) = (t_1(\theta), \dots, t_N(\theta))$, $\theta = (\theta_1, \dots, \theta_N) \in \Theta = \Theta_1 \times \cdots \times \Theta_N$, $v_i: X \to \R$, and $X$ is the set of feasible outcomes.
The pair $f = (d, t): \Theta \to X \times \R^N$ is called a social choice function.

Consider that each agent~$i$ reports a message $s_i(\theta_i) \in \Sigma_i$ to a social planner using a message function $s_i: \Theta_i \to \Sigma_i$ based on its type while attempting to minimize its own cost.
The social planner obtains the value of social choice function from the messages using an outcome function $g: \Sigma \to X \times \R^N$ such that $g \circ s(\theta) = f(\theta)$, where $\Sigma = \Sigma_1 \times \cdots \times \Sigma_N$ and $s(\theta) = (s_1(\theta_1), \dots, s_N(\theta_N))$.
Under this setting, the goal of the social planner is to design a mechanism consisting of $g$, $\Sigma$, and $s$ that implements $f$.

\begin{definition}[Incentive compatibility~\cite{Tanaka2017-zg}]
    A mechanism $M = (g, \Sigma, s)$ is \emph{incentive compatible} or \emph{implements} a social choice function $f$ in ex-post Nash equilibria if $g \circ s = f$ and $u_i(g(s_i(\theta_i), s_{-i}(\theta_{-i})); \theta_i) \le u_i(g(\sigma_i, s_{-i}(\theta_{-i})); \theta_i)$ hold for every $i \in \{1, \dots, N\}$, for all $\sigma_i \in \Sigma_i$, and for all $\theta \in \Theta$, where $\theta_{-i} \coloneqq (\theta_1, \dots, \theta_{i-1}, \theta_{i+1}, \dots, \theta_N)$, and $s_{-i}$ is defined in the same manner as $\theta_{-i}$.
\end{definition}

If a social choice function is implemented in ex-post Nash equilibria, no agent gains benefit by adopting a strategy $\sigma_i \ne s_i(\theta_i)$.
Thus, no agent is incentivized to deviate from the equilibrium provided all other agents play the equilibrium strategy.
Meanwhile, in the nature of algorithmic mechanism design, the outcome $(d(\theta), t(\theta))$ is sometimes given by a solution to an optimization algorithm.
In that case, the exact optimal outcome cannot be obtained in a finite iteration $n$ and must be approximated by a near-optimal one.
Then, a mechanism cannot be incentive compatible in general~\cite{Nisan2000-ri}.
To avoid this difficulty, the following notion is introduced.

\begin{definition}[Asymptotically incentive compatibility~\cite{Tanaka2017-zg}]
\label{def:aic}
    For every $n \in \N$, let $M^n = (g^n, \Sigma^n, s^n)$ be a mechanism.
    A sequence of mechanisms $\{M^n\}_{n \in \N}$ is \emph{asymptotically incentive compatible} or \emph{asymptotically implements} a social choice function $f$ in ex-post Nash equilibria if, for every $\varepsilon > 0$, there exists $n_0 \in \N$ such that $\lim_{n \to \infty} g^n \circ s^n = f$ and $u_i(g^n(s^n_i(\theta_i), s^n_{-i}(\theta_{-i})); \theta_i) \le u_i(g^n(\sigma^n_i, s^n_{-i}(\theta_{-i})); \theta_i) + \varepsilon$ hold for all $n \ge n_0$, for every $i \in \{1, \dots, N\}$, for all $\sigma_i^n \in \Sigma_i^n$, and for all $\theta \in \Theta$.
\end{definition}

This definition implies that if a mechanism given by an iterative algorithm is asymptotically incentive compatible, it converges to be incentive compatible as the number of iterations $n$ goes to infinity.
Moreover, even when $n$ is finite, the decrease of cost by manipulating messages is bounded by any small $\varepsilon$ for all agents if $n$ is sufficiently large.

\subsection{Secure multi-party computation}

Let $F_i : \{0, 1\}^\ast \times \cdots \times \{0, 1\}^\ast \to \{0, 1\}^\ast$ be a deterministic function such that $\yysf_i = F_i(\xxsf_1, \dots, \xxsf_N)$, where $\xxsf_i$ and $\yysf_i$ are private inputs and outputs of agent~$i$, and $\{0, 1\}^\ast$ is the set of binary sequences of any length.
Secure multi-party computation aims to compute a functionality $F = (F_1, \dots, F_N)$ on $(\xxsf_1, \dots, \xxsf_N)$ without revealing any information other than $\xxsf_i$ and $\yysf_i$ for each agent~$i$.
In an ideal world, the agents can achieve the objective by sending their inputs $\xxsf_i$ to a trusted third party that computes and returns $\yysf_i$ to each agent~$i$.
By contrast, in the real world, they jointly compute $F$ with communication because there is no trusted third party.
From this perspective, if all information that adversaries can obtain in the real world is also obtained in the ideal world, a protocol in the real world is considered at least as secure as one in the ideal world.

The simulation paradigm is a standard approach to formally define such security in multi-party computation.
In this approach, a protocol is considered secure under semi-honest adversaries if the \emph{view} of the adversaries are computationally indistinguishable from the information computed by their inputs and outputs.
Here, the view of agent~$i$ during an execution of a protocol $\Pi$ on a security parameter $\lambda$ and inputs $(\xxsf_1, \dots, \xxsf_N)$, denoted $\view{\Pi}_i(\lambda, \xxsf_1, \dots, \xxsf_N)$, is a tuple of $\xxsf_i$, internal random coins $\rrsf_i$ of the agent, and messages $\mmsf_i$ it has received~\cite{Lindell2016-gz}.
Note that the computational indistinguishability of two families of random variables implies that no polynomial-time algorithm can distinguish them~\cite{Lindell2016-gz}.

\begin{definition}[Secure multi-party computation~\cite{Goldreich2009-ct}]
\label{def:smpc}
    Let $F$ be a functionality that takes $(\xxsf_1, \dots, \xxsf_N)$ as input and outputs $(\yysf_1, \dots, \yysf_N)$.
    A protocol $\Pi$ \emph{$h$-privately computes} $F$ in the presence of semi-honest adversaries if there exist probabilistic polynomial-time algorithms (called simulators) $\Sim$ such that two families of random variables $\{ \Sim(1^\lambda, C, \{ \xxsf_i, \yysf_i \mid i \in C \}) \}_{\lambda \in \N, \xxsf_1, \dots, \xxsf_N \in \{0, 1\}^\ast}$ and $\{ \{ \view{\Pi}_i(\lambda, \xxsf_1, \dots, \xxsf_N) \mid i \in C \} \}_{\lambda \in \N, \xxsf_1, \dots, \xxsf_N \in \{0, 1\}^\ast}$ are computationally indistinguishable for every $C \subset \{1, \allowbreak \dots, N\}$ satisfying $|C| < h$, where $\xxsf_i$ are of equal length for all $i$.
\end{definition}

Intuitively, the definition implies that information obtained by a coalition of adversaries through the execution of a secure protocol can be simulated from their inputs and outputs.
In other words, the adversaries can learn nothing except for information given by their inputs and outputs.

\subsection{Homomorphic encryption and secret sharing}
\label{sec:he_ss}

Let $\M$ be a plaintext space, and $\C$ be a ciphertext space.
Additively homomorphic encryption, such as learning with errors (LWE) encryption, is an encryption scheme that allows addition over encrypted data.
That is, there exists a binary operation $\oplus: \C \times \C \to \C$ such that $\Dec(\sk, \Enc(\pk, m_1) \oplus \Enc(\pk, m_2)) = m_1 + m_2 \in \M$ for all $m_1, m_2 \in \M$, where $\Enc$ is an encryption algorithm, $\Dec$ is a decryption algorithm, $\pk$ is a public key, and $\sk$ is a secret key.
With the homomorphic addition $\oplus$, a binary operation $\odot: (\ct, n) \mapsto \ct \oplus \dots \oplus \ct$ is defined for $\ct = \Enc(\pk, m)$ and $n \in \N$ such that $\Dec(\sk, \ct \odot n) = mn \in \M$.
Furthermore, we assume that additively homomorphic encryption satisfies semantic security~\cite{Goldreich2009-ct}.
This implies that the encryption of a plaintext gives no information on the plaintext to a polynomial-time adversary.

Additive secret sharing over $\Z_q \coloneqq \{0, 1, \allowbreak \dots, q - 1\}$ is a cryptographic technique to store a secret distributedly.
It splits a message $m \in \Z_q$ into $n$ shares by a share generation algorithm $(\sss_1, \dots, \sss_n) \gets \Share(m, n)$, where $\sss_1, \dots, \sss_{n - 1}$ are sampled from $\Z_q$ uniformly at random, and $\sss_n = m - \sum_{i = 1}^{n - 1} \sss_i \bmod q$.
The message can be recovered by a reconstruction algorithm as $\Reconst(\sss_1, \dots, \sss_n) = \sum_{i = 1}^n \sss_i \bmod q$.
Correctness of additive secret sharing is obvious, namely $\Reconst(\Share(m, n)) = m$ for all $m \in \Z_q$ and $n \ge 2$.
Moreover, any $n - 1$ shares are uniformly at random over $\Z_q$ and independent of $m$ by construction.

\section{Privacy-preserving Average Consensus}
\label{sec:consensus}

Using the cryptographic tools in \secref{sec:he_ss}, we present \proref{pro:consensus} that computes
\begin{equation}
    \vv_i(k) \coloneqq \sum{}_{j \in \Nin_i} \vv_{ij}(k), \quad \vv_{ij}(k) \coloneqq \ww_{ij} \xx_j(k)
    \label{eq:aggregation}
\end{equation}
over encrypted data.
Here, we focus on securely computing the second term in \eqref{eq:control} because agent~$i$ can locally compute the first term.
For the sake of simplicity, assume that $\ww_{ij}$ are rational numbers for all $i, j \in V$, and the plaintext space is $\M = \Z_q$ with a large prime $q$.
These assumptions are reasonable in practice because a real-valued weight can be approximated by a rational number with any desired precision to inherit the stability and performance of the original control, and $q$ can be chosen freely.

\begin{figure}[!t]
    \begin{algorithm}[H]
        \caption{Privacy-preserving average consensus}
        \label{pro:consensus}
        \begin{algorithmic}[1]
            \Require $n$, $G$, $\Delta_\ww$, $\Delta_\xx$, $\ww_{ii}$, $\ww_{ij}$, $\xx_{i,0}$, $\pk_i$, $\sk_i$
            \Ensure $\xx_i(n)$
            \State Send $\{ \ct_{\sss, ji}(k) \}_{k=0,\, j \in \Nout_i}^{n-1}$ to agent~$i$ \Comment{Supervisor}
            \State Broadcast $\{ \ct_{\ww, ij} \}_{j \in \Nin_i}$ to all agents \Comment{Agent~$i$}
            \ForAll{$k = 0, \dots, n - 1$} \Comment{Agent~$i$}
                \State Send $\ct_{\vv, ji}(k)$ to $j \in \Nout_i$
                \State $\vv_i(k) \gets \Delta [ \Reconst( (\Dec(\sk_i, \ct_{\vv, ij}(k)))_{j \in \Nin_i} ) ]_q$
                \State $\uu_i(k) \gets \ww_{ii} \xx_i(k) + \vv_i(k)$, $\xx_i(k + 1) \gets \xx_i(k) + \uu_i(k)$
            \EndFor
        \end{algorithmic}
    \end{algorithm}
    \vspace{-10mm}
\end{figure}

In the offline phase, the supervisor generates and sends $\{ \ct_{\sss, ji}(k) \gets \Enc(\pk_j, \sss_{ji}(k)) \}_{k = 0, \, j \in \Nout_i}^{n - 1}$ to agent~$i$, where $(\sss_{ij}(k))_{j \in \Nin_i} \gets \Share(0, |\Nin_i|)$ for every $k \in \N_0$.
Simultaneously, agent~$i$ encrypts its weights as $\ct_{\ww, ij} \gets \Enc(\pk_i, \tilde{\ww}_{ij})$ and broadcasts $\{\ct_{\ww, ij}\}_{j \in \Nin_i}$ to all agents, where $\tilde{\ww}_{ij} \coloneqq \bar{\ww}_{ij} \bmod q$, $\bar{\ww}_{ij} \coloneqq \Delta_\ww^{-1} \ww_{ij}$, and $\Delta_\ww > 0$.
Note that since $\ww_{ij}$ are rational numbers, there exists $\Delta_\ww$ such that $\bar{\ww}_{ij} \in \Z$ for all $i \in V$ and for all $j \in \Nin_i$.
Additionally, the broadcast process can be performed by message passing via the supervisor.

In the online phase, agent~$i$ encodes its state as $\tilde{\xx}_i(k) \coloneqq \bar{\xx}_i(k) \bmod q$, where $\bar{\xx}_i(k) \coloneqq \lfloor \Delta_\xx^{-1} \xx_i(k) \rceil$, $\Delta_\xx > 0$, and $\lfloor \cdot \rceil$ represents a rounding of a real number into the nearest integer.
It then computes and sends $\ct_{\vv, ji}(k) = \ct_{\ww, ji} \odot \tilde{\xx}_i(k) \oplus \ct_{\sss, ji}(k)$ to agent~$j \in \Nout_i$.
Upon receiving $\ct_{\vv, ij}(k)$, agent~$i$ computes $\vv_i(k) = \Delta [ \Reconst( (\Dec(\sk_i, \ct_{\vv, ij}(k)))_{j \in \Nin_i} ) ]_q$, where $\Delta = \Delta_\ww \Delta_\xx$, and $[z]_q \coloneqq z - \lfloor \frac{z + q / 2}{q} \rfloor q$ is the minimal residue of $z$ modulo $q$.
The agent then updates its state as $\xx_i(k + 1) = \xx_i(k) + \ww_{ii} \xx_i(k) + \vv_i(k)$.
Note that although the resultant $\vv_i(k)$ includes a quantization error due to the rounding process, we ignore it in the following because it can be arbitrarily small by choosing sufficiently small $\Delta_\xx$.
Consequently, we obtain the proposition below.

\begin{proposition}
\label{prop:consensus}
    The outputs of \proref{pro:consensus} achieve $\lim_{n \to \infty} \xx_i(n) = \frac{1}{N} \sum_{j=1}^N \xx_{j,0}$ for all $i \in V$ if, for every $k = 0, \dots, n$, it holds that $| \bar{\vv}_{ij}(k) | < q / 2$ and $| \bar{\vv}_i(k) | < q / 2$ for all $i \in V$ and for all $j \in \Nin_i$, where $\bar{\vv}_i(k) = \sum_{j \in \Nin_i} \bar{\vv}_{ij}(k)$ and $\bar{\vv}_{ij}(k) = \bar{\ww}_{ij} \bar{\xx}_j(k)$.
\end{proposition}

\begin{proof}
    The claim holds from that $\vv_i(k)$ in the protocol is equivalent to \eqref{eq:aggregation} due to the homomorphism and correctness of additively homomorphic encryption and secret sharing.
\end{proof}

\propref{prop:consensus} shows that \proref{pro:consensus} achieves average consensus when $| \bar{\vv}_{ij}(k) | < q / 2$ and $| \bar{\vv}_i(k) | < q / 2$ hold, where recall that $q$ can be chosen freely to satisfy the conditions.
However, the protocol guarantees nothing about whether the agents follow it faithfully.
Indeed, the protocol outputs deviate from an average value if an agent misreports its initial state or modifies its input.
This problem will be solved later based on mechanism design theory.

The rest of this section demonstrates the security of \proref{pro:consensus} under semi-honest adversaries.
In what follows, the supervisor is regarded as the $0$th party.
The following assumptions are also made to specify an attack scenario.

\begin{assumption}
\label{asm:mpc}
    Assume the following conditions.
    \begin{itemize}
        \item The supervisor does not collude with any agent.
        \item Every agent has more than two input neighbors.
        \item $n$, $G$, $\Delta_\ww$, $\Delta_\xx$, $\M$, $\C$, and $\{\pk_i\}_{i \in V}$ are public.
    \end{itemize}
\end{assumption}

Note that the first and second assumptions are necessary in our scenario.
If agent~$i$ colludes with the supervisor, it can identify $\tilde{\xx}_j(k)$ for all $j \in \Nin_i$ and for all $k \in \N_0$ because the supervisor has all shares $\sss_{ij}(k)$.
Additionally, if $| \Nin_i | = 1$, agent~$i$ can easily identify agent~$j$'s state as $\xx_j(k) = \ww_{ij}^{-1} \vv_i(k)$ from \eqref{eq:aggregation}.
With the assumption, the lemma below shows the security of \proref{pro:consensus} under semi-honest adversaries less than $\min_i |\Nin_i|$.

\begin{lemma}
\label{lem:consensus}
    Let $\xxsf_i = (\ww_{ii}, \{\ww_{ij}\}_{j \in \Nin_i}, \xx_{i,0}, \sk_i)$ and $\yysf_i = \{ \xx_i(k) \}_{k=0}^n$ for $i \in V$ and $n \in \N$.
    \proref{pro:consensus} $h$-privately computes functionality $(\Lambda, \yysf_1, \dots, \yysf_N) = F(\Lambda, \xxsf_1, \dots, \xxsf_N)$ in the presence of semi-honest adversaries under \asmref{asm:mpc}, where $h = \min_i |\Nin_i|$, and $\Lambda$ is the empty string.
\end{lemma}

\begin{proof}
    Let $\Pi$ be \proref{pro:consensus}.
    Under \asmref{asm:mpc}, this proof constructs simulators $\Sim$ that satisfy the condition in \defref{def:smpc} for adversarial agents ($C \subset V$) and the supervisor ($C = \{0\}$) separately.
    From Corollary~2 in~\cite{Canetti2000-xq}, a sequential composition of protocols $\Pi_1, \dots, \Pi_T$, which respectively and privately compute functionalities $F_1, \dots, F_T$, privately computes a composition of the functionalities in the presence of semi-honest adversaries.
    Thus, the proof suffices only for $n = 1$ because $\Pi$ with any $n \in \N$ can be realized by sequentially repeating $\Pi$ with $n = 1$.
    
    \emph{Simulator for agents:}
    The view of agent~$i$ is given by $\view{\Pi}_i(\lambda, \Lambda, \xxsf_1, \dots, \xxsf_N) \!=\! (\xxsf_i, \rrsf_i, \mmsf_i)$, where $\rrsf_i \!=\! \{r_{ij}\}_{j \in \Nin_i}$ are seeds for random numbers used in the encryption of $\tilde{\ww}_{ij}$, and $\mmsf_i = (\mmsf_{i, \sss}, \mmsf_{i, \ww}, \mmsf_{i, \vv}) = (\{\ct_{\sss, ji}(0)\}_{j \in \Nout_i}, \allowbreak \{\ct_{\ww, ji}\}_{i \in V, j \in \Nout_i}, \{\ct_{\vv, ij}(0)\}_{j \in \Nin_i})$.
    Construct a simulator $\Sim$ as follows:
    1) Generate seeds $\hat{\rrsf}_i$ of the equal length as $\rrsf_i$ uniformly at random for all $i \in C$.
    2) Sample $\hat{\ct}_{\sss, ji}$ from $\C$ uniformly at random for all $i \in C$ and for all $j \in \Nout_i$.
    3) For all $i \in V$ and for all $j \in \Nout_i$, compute $\hat{\ct}_{\ww, ji} \allowbreak \gets \Enc(\pk_j, \tilde{\ww}_{ji})$ if $j \in C$; otherwise sample $\hat{\ct}_{\ww, ji}$ from $\C$ uniformly at random.
    4) For all $i \in C$ and for all $j \in \Nin_i$, compute $\hat{\ct}_{\vv, ij} \gets \Enc(\pk_i, \tilde{\ww}_{ij} \tilde{\xx}_{j,0} \bmod q) \oplus \hat{\ct}_{\sss, ij}$ if $j \in C$; otherwise sample $\hat{\ct}_{\vv, ij}$ from $\C$ uniformly at random.
    5) Let $\hat{\mmsf}_i \!=\! (\hat{\mmsf}_{i, \sss}, \hat{\mmsf}_{i, \ww}, \hat{\mmsf}_{i, \vv}) \!=\! (\{ \hat{\ct}_{\sss, ji} \}_{j \in \Nout_i}, \! \{ \hat{\ct}_{\ww, ji} \}_{i \in V, j \in \Nout_i}, \allowbreak \{ \hat{\ct}_{\vv, ij} \}_{j \in \Nin_i})$.
    Output $\{ (\xxsf_i, \hat{\rrsf}_i, \hat{\mmsf}_i) \mid i \in C \}$.

    By construction, it holds that $\Dec(\sk_i, \ct_{\vv, ij}(0)) = \tilde{\ww}_{ij} \tilde{\xx}_{j,0} + \sss_{ij}(0) \bmod q$ and $\Dec(\sk_i, \hat{\ct}_{\vv, ij}) = \tilde{\ww}_{ij} \tilde{\xx}_{j,0} + \hat{\sss}_{ij} \bmod q$ for all $i \in C$ and for all $j \in \Nin_i$, where $\hat{\sss}_{ij}$ is uniformly random over $\Z_q$.
    From the randomness of additive secret sharing, if $|C| < \min_i |\Nin_i|$, $\{ \{\hat{\sss}_{ji}\}_{j \in \Nout_i} \mid i \in C\}$ and $\{ \{\Dec(\sk_i, \hat{\ct}_{\vv, ij})\}_{j \in \Nin_i} \mid i \in C \}$ have the same distribution as $\{ \{\sss_{ji}(0)\}_{j \in \Nout_i} \mid i \in C\}$ and $\{ \{\Dec(\sk_i, \ct_{\vv, ij}(0))\}_{j \in \Nin_i} \mid i \in C \}$, respectively.
    Hence, semantic security of additively homomorphic encryption implies that $\{ (\hat{\mmsf}_{i, \sss}, \hat{\mmsf}_{i, \vv}) \mid i \in C \}$ and $\{ (\mmsf_{i, \sss}, \mmsf_{i, \vv}) \mid i \in C \}$ are computationally indistinguishable even given $\{ (\xxsf_i, \rrsf_i) \mid i \in C\}$.
    It also implies that $\{ \mmsf_{i, \ww} \mid i \in C \}$ and $\{ \hat{\mmsf}_{i, \ww} \mid i \in C \}$ are computationally indistinguishable even given $\{ (\xxsf_i, \rrsf_i) \mid i \in C\}$, and $\{ \mmsf_{i, \ww} \mid i \in C \}$ is conditionally independent of $\{ (\mmsf_{i, \sss}, \mmsf_{i, \vv}) \mid i \in C \}$ given $\{ \xxsf_i \mid i \in C\}$.
    Consequently, the condition in \defref{def:smpc} holds.

    \emph{Simulator for the supervisor:}
    The construction is obvious because the supervisor receives no message.
\end{proof}

Note that the supervisor in \lemref{lem:consensus} takes and outputs the empty string, which means that it gives no input and receives no output in the protocol.
This is because, to assist agents' computation, it just sends the encryption of shares in the offline phase.

\section{Distributed Mechanism for Privacy-Preserving Average Consensus}
\label{sec:mechanism}

In this section, we assume a rational adversary model instead of a semi-honest adversary model.

\begin{definition}
\label{def:adversary}
    Agent $i$ is a \emph{rational adversary} if it performs $\min_{s_i(\theta_i)} u_i(d(\theta), t(\theta); \theta_i)$ and attempts to learn information about other agents from one's view, where $s_i(\theta_i) = \{\sigma_{i, k} \mid k = 0, \dots, n - 1\}$, and $\sigma_{i, k}$ are outgoing messages that agent~$i$ sends to its output neighbors at time $k$.
\end{definition}

The rational adversaries formulated in the definition are allowed to cooperate with each other to learn the private information of honest agents.
Meanwhile, they are supposed to minimize their own costs individually.
This is a natural setting because, in practice, adversaries would have conflicting objectives (i.e., minimizing each cost), even if they agree to compromise the privacy of honest agents.

Our objective is to design a privacy-preserving protocol for providing a mechanism that implements a social choice function with decision rule $d(\theta) = (\frac{1}{N} \sum_{i = 1}^N \xx_{i,0}, \allowbreak \dots, \frac{1}{N} \sum_{i = 1}^N \xx_{i,0})$ under rational adversaries.
Here, computing average $\frac{1}{N} \sum_{i = 1}^N \xx_{i,0}$ is equivalent to minimizing $\sum_{i = 1}^N (\zz_i - \theta_i)^2$ for $(\zz_1, \dots, \zz_N)$ with $\xx_{i,0} = \theta_i$~\cite{Rabbat2005-bc}.
This fact suggests that the average consensus problem can be regarded as a mechanism design problem with individual costs $u_i(d(\theta), t(\theta); \theta_i) = (\zz_i - \theta_i)^2 + t_i(\theta)$, where $d(\theta) = (\zz_1, \dots, \zz_N)$.

We propose \proref{pro:mechanism} based on the above observation.
In the offline phase, the supervisor generates and sends $\{ \ct_{\trm, ji} \gets \Enc(\pk_j, \trm_{ji}) \}_{j \in V \setminus \{i\}}$ to each agent~$i \in V$, where $(\trm_{ij})_{j \in V \setminus \{i\}} \gets \Share(0, N - 1)$ for every $i \in V$.
Then, the protocol invokes \proref{pro:consensus} with $\xx_{i,0} = \theta_i$.
After the online phase of \proref{pro:consensus}, agent~$1$ encrypts its state as $\{ \ct_{\xx, i1} \gets \Enc(\pk_i, \tilde{\xx}_1(n)) \}_{i \in V}$ and broadcasts them to all agents, where $\tilde{\xx}_1(n) = \bar{\xx}_1(n) \bmod q$ is the encoded terminal state of agent~$1$.
Simultaneously, agent~$i$ broadcasts $\{ \ct_{v, ji} \gets \Enc(\pk_j, \tilde{v}_i) \oplus \ct_{\trm, ji} \}_{j \in V \setminus \{i\}}$ to all agents, where $\tilde{v}_i = \lfloor \Delta_\xx^{-1} (\xx_i(n) - \theta_i)^2 \rceil \bmod q$.
Then, agent~$i$ obtains the social decision as $d_i(\theta) = \Delta_\xx [ \Dec(\sk_i, \ct_{\xx, i1}) ]_q$ and $t_i(\theta) = \Delta_\xx [ \Reconst( (\Dec(\sk_i, \ct_{v, ij}))_{j \in V \setminus \{i\}} ) ]_q$.
Consequently, for a sufficiently large $n$, the social outcome is given as $d_i(\theta) = \xx_1(n) \approx \frac{1}{N} \sum_{j = 1}^N \theta_j$ with $t_i(\theta) = \sum_{j \ne i} (d_j(\theta) - \theta_j)^2$.

\begin{proposition}
\label{prop:mechanism}
    Suppose that, for every $k = 0, \dots, n$ and $n \in \N$, $| \bar{\vv}_{ij}(k) | < q / 2$, $| \bar{\vv}_i(k) | < q / 2$, $|\bar{\xx}_1(n)| < q / 2$, and $\sum_{j \ne i} \lfloor \Delta_\xx^{-1} (\xx_i(n) - \theta_i)^2 \rceil < q / 2$ hold for all $i \in V$ and for all $j \in \Nin_i$, where $\bar{\vv}_{ij}$ and $\bar{\vv}_i$ are as in \propref{prop:consensus}.
    Let $M^n = (g^n, \Sigma^n, s^n)$ be a mechanism provided by \proref{pro:mechanism}.
    The sequence of mechanisms $\{M^n\}_{n \in \N}$ asymptotically implements the social choice function $f = (d, t)$ given by $d(\theta) = ( d_1(\theta), \dots, d_N(\theta) )$, $d_i(\theta) = \frac{1}{N} \sum_{j = 1}^N \theta_j$, $t(\theta) = ( t_1(\theta), \dots, t_N(\theta) )$, and $t_i(\theta) = \sum_{j \ne i} (d_j(\theta) - \theta_j)^2$ in ex-post Nash equilibria.
\end{proposition}

\begin{proof}
    Let $\hat{d}(\theta) = (\hat{d}_1(\theta), \dots, \hat{d}_N(\theta))$ and $\hat{t}(\theta) = (\hat{t}_1(\theta), \dots, \hat{t}_N(\theta))$ be decision and transfer rules computed by the protocol.
    By construction, it follows that $\hat{d}_i(\theta) = \Delta_\xx [ \Dec(\sk_i, \ct_{\xx, i1}) ]_q = \xx_1(n)$ and $\hat{t}_i(\theta) = \Delta_\xx [ \Reconst( (\tilde{v}_j + \trm_{ij})_{j \in V \setminus \{i\}} ) ]_q = \sum_{j \ne i} (\xx_j(n) - \theta_j)^2$ for every $i \in V$.
    \propref{prop:consensus} implies $\lim_{n \to \infty} g^n \circ s^n = f$ because $\lim_{n \to \infty} \xx_i(n) = \frac{1}{N} \sum_{j = 1}^N \theta_j$ for all $i \in V$.
    Therefore, the claim follows from Proposition~2 in~\cite{Tanaka2013-bl}.
\end{proof}

\begin{figure}[!t]
    \begin{algorithm}[H]
        \caption{Privacy-preserving distributed mechanism}
        \label{pro:mechanism}
        \begin{algorithmic}[1]
            \Require $n$, $G$, $\Delta_\ww$, $\Delta_\xx$, $\ww_{ii}$, $\ww_{ij}$, $\theta_i$, $\pk_i$, $\sk_i$
            \Ensure $d_i(\theta)$, $t_i(\theta)$
            \State Send $\{ \ct_{\trm, ji} \}_{j \in V \setminus \{i\}}$ to agent~$i$ \Comment{Supervisor}
            \State Invoke \proref{pro:consensus} with $\xx_{i,0} = \theta_i$ for all $i \in V$
            \State Broadcast $\{ \ct_{\xx, i1} \}_{i \in V}$ to all agents \Comment{Agent~$1$}
            \State Broadcast $\{ \ct_{v, ji} \}_{j \in V \setminus \{i\}}$ to all agents \Comment{Agent~$i$}
            \State $d_i(\theta) \gets \Delta_\xx [ \Dec(\sk_i, \ct_{\xx, i1}) ]_q$
            \State $t_i(\theta) \gets \Delta_\xx [ \Reconst( (\Dec(\sk_i, \ct_{v, ij}))_{j \in V \setminus \{i\}} ) ]_q$
        \end{algorithmic}
    \end{algorithm}
    \vspace{-10mm}
\end{figure}

\propref{prop:mechanism} implies that, as $n \to \infty$, all agents report their types honestly, i.e., $\theta_i = \xx_{i,0}$, and follow \proref{pro:mechanism} faithfully.
Then, the behavior of rational adversaries in \proref{pro:mechanism} is equivalent to semi-honest adversaries.
The theorem below demonstrates the security of \proref{pro:mechanism} in the same manner as \lemref{lem:consensus}.

\begin{theorem}
\label{thm:mechanism}
    Let $\xxsf_i = (\ww_{ii}, \{\ww_{ij}\}_{j \in \Nin_i}, \theta_i, \sk_i)$ and $\yysf_i = (\{ \xx_i(k) \}_{k=0}^n, d_i(\theta), t_i(\theta))$ for all $i \in V$.
    \proref{pro:mechanism} $h$-privately computes functionality $(\Lambda, \yysf_1, \dots, \yysf_N) = F(\Lambda, \allowbreak \xxsf_1, \dots, \xxsf_N)$ in the presence of semi-honest adversaries under \asmref{asm:mpc}, where $h$ and $\Lambda$ are as in \lemref{lem:consensus}.
\end{theorem}

\begin{proof}
    Let $\Pi$ be \proref{pro:mechanism}, and $\Pi'$ be a protocol excluding line~2 from $\Pi$.
    The claim follows from Corollary~2 in~\cite{Canetti2000-xq} and \lemref{lem:consensus} by compositing \proref{pro:consensus} and $\Pi'$ if $\Pi'$ $h$-privately computes $F$ in the presence of semi-honest adversaries.
    This proof constructs simulators $\Sim$ for $\Pi'$.
    
    \emph{Simulator for agents:}
    The view of agent~$i$ is given by $\view{\Pi'}_i(\lambda, \Lambda, \xxsf_1, \dots, \xxsf_N) = (\xxsf_i, \rrsf_i, \mmsf_i)$, where $\rrsf_i = \{r_{ij}\}_{j \in V \setminus \{i\}}$ for $i \ne 1$, $\rrsf_1 =  (\{r_\ell\}_{\ell \in V}, \{r_{1j}\}_{j \in V \setminus \{1\}})$, and $\mmsf_i = (\mmsf_{i, \trm}, \mmsf_\xx, \mmsf_v) = (\{\ct_{\trm, ji}\}_{j \in V \setminus\{i\}}, \{\ct_{\xx, \ell 1}\}_{\ell \in V}, \allowbreak \{\ct_{v, j \ell}\}_{\ell \in V, j \in V \setminus \{\ell\}}\})$.
    $\{r_\ell\}_{\ell \in V}$ and $\{r_{ij}\}_{j \in V \setminus \{i\}}$ are seeds for random numbers used in the encryption of $\tilde{\xx}_1(n)$ and $\tilde{v}_i$, respectively.
    Construct a simulator $\Sim$ as follows:
    1) Generate seeds $\hat{\rrsf}_i$ of the equal length as $\rrsf_i$ uniformly at random for all $i \in C$.
    2) Compute $\hat{\ct}_{\trm, ji} \gets \Enc(\pk_j, \hat{\trm}_{ji})$ with $(\hat{\trm}_{ij})_{j \in V \setminus \{i\}} \gets \Share(0, N - 1)$ for all $i \in V$.
    3) For all $\ell \in V$, compute $\hat{\ct}_{\xx, \ell 1} \gets \Enc(\pk_\ell, \Delta_\xx^{-1} d_i(\theta) \bmod q)$ with some $i \in C$.
    4) For all $i \in V$ and for all $j \in V \setminus \{i\}$, compute $\hat{\ct}_{v, ji} \gets \Enc(\pk_j, \allowbreak \tau_{ji}) \oplus \hat{\ct}_{\trm, ji}$, where $\tau_{ji} = \tilde{v}_i$ if $i \in C$, $\tau_{ji} = 0$ if $i \in V \setminus C$ and $j \in (V \setminus \{i\}) \setminus C$, and $(\tau_{ji})_{i \in V \setminus C} \gets \Share(\Delta_\xx^{-1} t_j(\theta) - \sum_{\ell \in C \setminus \{j\}} \tilde{v}_\ell, | V \setminus C |)$ if $i \in V \setminus C$ and $j \in C$.
    5) Let $\hat{\mmsf}_i = (\hat{\mmsf}_{i, \trm}, \hat{\mmsf}_\xx, \hat{\mmsf}_v) = (\{\hat{\ct}_{\trm, ji}\}_{j \in V \setminus \{i\}}, \{\hat{\ct}_{\xx, \ell 1}\}_{\ell \in V}, \allowbreak \{\hat{\ct}_{v, \ell j}\}_{\ell \in V, j \in V \setminus \{\ell\}}\})$.
    6) Output $\{ (\xxsf_i, \hat{\rrsf}_i, \hat{\mmsf}_i) \mid i \in C \}$.

    By construction, $\{ (\mmsf_{i, \trm}, \mmsf_\xx) \!\mid\! i \in C \}$ and $\{ (\hat{\mmsf}_{i, \trm}, \hat{\mmsf}_\xx) \!\mid\! i \in C \}$ have the same distribution.
    For all $i \in V$ and for all $j \in V \setminus \{i\}$, $\ct_{v, ji}$ is computationally indistinguishable from $\hat{\ct}_{v, ji}$ due to semantic security.
    Furthermore, for all $i \in C$, it holds that $\Reconst( (\Dec(\sk_i, \hat{\ct}_{v, ij}) )_{j \in V \setminus \{i\}} ) = \sum_{j \in C \setminus \{i\}} \tau_{ij} + \sum_{j \in V \setminus C} \tau_{ij} = (\sum_{j \in C \setminus \{i\}} \tilde{v}_j) + (\Delta_\xx^{-1} t_i(\theta) - \sum_{\ell \in C \setminus \{j\}} \tilde{v}_\ell) = \Delta_\xx^{-1} t_i(\theta)$, which means that $\{ \mmsf_v \mid i \in C \}$ and $\{ \hat{\mmsf}_v \mid i \in C \}$ are computationally indistinguishable even given $\{ (\xxsf_i, \rrsf_i) \mid i \in C \}$.
    Therefore, the condition in \defref{def:smpc} holds.

    \emph{Simulator for the supervisor:}
    The construction is obvious because the supervisor receives no message.
\end{proof}

Combining with \propref{prop:mechanism} and \thmref{thm:mechanism}, as $n \to \infty$, the security of \proref{pro:mechanism} is guaranteed in the sense of \defref{def:smpc} with $h = \min_i |\Nin_i|$ under \asmref{asm:mpc} even for rational adversaries.
Note that, for finite $n$, rational adversaries are not equivalent to semi-honest ones, and then they might not completely follow the proposed protocol.
However, according to \defref{def:aic} and \propref{prop:mechanism}, the decrease of adversaries' costs by deviating from the protocol is bounded by any small value if $n$ is sufficiently large.
In this light, their behavior can be made arbitrarily close to semi-honest ones by choosing large $n$.

\begin{remark}
    A mechanism is \emph{(weakly) budget balanced} if $\sum_{i = 1}^N t_i(\theta) = 0$.
    Moreover, it is \emph{individually rational} if $u_i(d(\theta), t(\theta); \theta_i) \le 0$ for every $i \in V$.
    The mechanism provided by \proref{pro:mechanism} is neither budget balanced nor individually rational, although it is asymptotically incentive compatible.
    Further development of the proposed protocol to satisfy the properties is future work. 
\end{remark}

\begin{remark}
    The supervisor must verify that all agents pay $t_i(\theta)$ correctly after executing \proref{pro:mechanism}.
    This is not straightforward because the supervisor does not know the exact values of $t_i(\theta)$ due to encryption.
    Nevertheless, the value of a sum of $t_i(\theta)$ can be verified without compromising privacy as follows.
    Suppose $t_i'(\theta)$ is the value that agent~$i$ actually paid.
    Let $v_i(\theta) = (\xx_i(n) - \theta_i)^2$ and $u_i(\theta) =  v_i(\theta) + t_i(\theta)$.
    It follows that $u_i(\theta) = \sum_{j=1}^N v_j(\theta)$ and $\sum_{i=1}^N t_i(\theta) = (N - 1) \sum_{j=1}^N v_j(\theta)$.
    Combining these equations, we obtain $\sum_{i=1}^N t_i(\theta) = (N - 1) u_i(\theta)$.
    Therefore, the supervisor can check whether the value of $\sum_{i=1}^N t_i'(\theta)$ is correct by asking each agent if $(N - 1)^{-1} \sum_{i=1}^N t_i'(\theta)$ is equal to $u_i(\theta)$.
\end{remark}

\section{Numerical Examples}
\label{sec:examples}

This section presents numerical examples using the LWE encryption.
We used ECLib~\cite{eclib} and lattice-estimator~\cite{Albrecht2015-pm} to implement the proposed protocols and the encryption scheme with $\lambda = 128$~bit security.

Let $n = 30$, $N = 5$, $\Delta_\ww = 0.1$, $\Delta_\xx = 0.01$, $\epsilon = 0.1$, and $A = [0 \ 0 \ 0 \ 1 \ 1; \ 1 \ 0 \ 0 \ 1 \ 1; \ 0 \ 1 \ 0 \ 1 \ 0; \ 0 \ 1 \ 0 \ 0 \ 1; \ 0 \ 1 \ 1 \ 0 \ 0]$.
\figref{fig:example}\subref{fig:attack_free} depicts the state trajectories of the agents during the execution of \proref{pro:consensus} and \proref{pro:mechanism} with $(\theta_1, \dots, \theta_5) = (3, 2, 1, 0, -1)$.
The cost of agent~$2$ was $u_2(d(\theta), t(\theta); \theta_2) = 10.65$, where $v_2(d(\theta); \theta_2) = 1.85$ and $t_2(\theta) = 8.80$.
\figref{fig:example}\subref{fig:attacked} shows the state trajectories when agent~$2$ stayed in the same state (i.e., $\xx_2(k) = 2$) by violating \proref{pro:consensus}.
In that case, although $v_2$ was reduced to $0$, $t_2$ was increased to $14.43$, thereby increasing $u_2$ to $14.43$.
This implies that the agent would never behave in such a manner as long as it is rational.

\begin{figure}[t]
    \vspace{-2mm}
    \centering
    \subfigure[No agent deviated.]{\includegraphics[scale=.5]{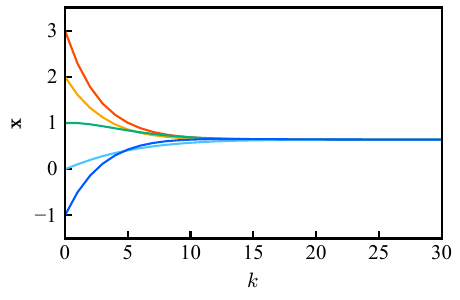}\label{fig:attack_free}}\hspace{3mm}%
    \subfigure[Agent~$2$ deviated.]{\includegraphics[scale=.5]{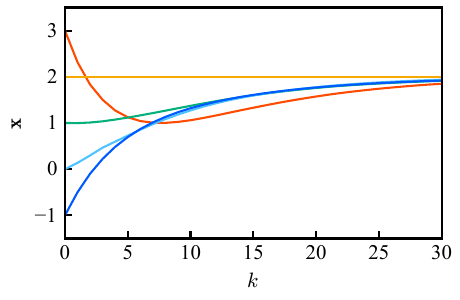}\label{fig:attacked}}
    \caption{State trajectories during the execution of the proposed protocols.}
    \label{fig:example}
    \vspace{-6mm}
\end{figure}

\section{Conclusions}
\label{sec:conclusions}

We proposed a privacy-preserving protocol to solve average consensus problems for rational and strategic agents.
The proposed protocol provides a distributed mechanism that incentivizes such agents to implement intended behavior faithfully and protects the privacy of agents using additively homomorphic encryption and additive secret sharing.
Combining the mechanism and cryptographic primitives, the proposed protocol fulfills security under rational adversaries rather than semi-honest adversaries.
The results of this study will be generalized to other cooperative control tasks.

\bibliographystyle{IEEEtran}
\bibliography{reference}

\end{document}